\documentclass[10pt,rqno, a4paper]{article}

\usepackage{titlesec}
\titleformat*{\section}{\sc\centering\large} 
\titleformat*{\subsection}{\bf} 
\titleformat*{\subsubsection}{\it} 

\usepackage{amsmath}
\usepackage{amssymb}
\usepackage{amsfonts}
\usepackage{setspace}
\usepackage{latexsym}
\usepackage{mathrsfs}
\usepackage{epsfig}
\usepackage{amsthm}
\usepackage{color}

\usepackage[margin=2cm]{caption}

\newcommand{\be}{\begin{equation}}
\newcommand{\ee}{\end{equation}}

\newcommand{\vs}{\vspace{0.2cm}}

\usepackage{fancyhdr}
\newtheorem{Theorem}{Theorem}
\newtheorem{Remark}[Theorem]{Remark}

\newtheorem{Proposition}[Theorem]{Proposition}

\newcommand{\dist}{d} 

\newcommand{\diss}{\displaystyle}

\newcommand{\length}{L} 

\newcommand{\Sa}{{\rm S}^{1}} 

\newcommand{\qM}{S} 

\newcommand{\hqg}{q} 

\usepackage{tocloft} 

\usepackage{setspace, tocloft}

\allowdisplaybreaks

\begin{document}

\thispagestyle{empty}

\begin{center}
{\Large\bf A complete classification of ${\bf \Sa}$-symmetric
\vs

static vacuum black holes}

\vs\vs\vs

{\sc Mart\'in Reiris}

{mreiris@cmat.edu.uy}

\vs

{\sc Javier Peraza}

{jperaza@cmat.edu.uy} 

\vspace{.4cm}

{\it Centro de Matem\'atica, Universidad de la Rep\'ublica} 

{\it Montevideo, Uruguay}
\vs\vs

\end{center}

\begin{abstract}
In a seminal paper of 1917, H. Weyl presented a remarkable reduction of the static axisymmetric vacuum Einstein equations, serving as a relatively straightforward technique to generate and explore new solutions. Weyl's reduction was used by Myers in 1987, and independently by Korotkin-Nicolai in 1994, to construct a new family of static and axisymmetric solutions with compact  non-empty horizon, however with non-trivial topology and asymptotically Kasner. This family, together with the Schwarzschild and the Boost families, remained until now as the only known $\Sa$-symmetric static black hole solutions, namely, (metrically complete) $\Sa$-symmetric static vacuum solutions with compact and non-empty horizon. In this article we prove that, indeed, these three families exhaust all the examples of $\Sa$-symmetric static vacuum black holes.
\end{abstract}

\section*{Introduction}

The classification of the static solutions of the vacuum Einstein equations is a central problem in Mathematical Relativity and in Geometry. In this article we will treat static solutions at the ``initial data level'', namely we will deal with a Riemannian three-manifold $(\Sigma;g)$ and a lapse function $N$, that we will assume is positive in the interior $\Sigma^{\circ}=\Sigma\setminus \partial \Sigma$ of $\Sigma$. Given the data $(\Sigma; g,N)$ the static vacuum spacetime $({\bf M}; {\bf g})$ is constructed as,
\be
{\bf M}=\Sigma\times \mathbb{R},\quad {\bf g}=-N^{2}dt^{2}+g,
\ee
and the vacuum condition ${\bf Ric}=0$ is equivalent to the static vacuum Einstein equations,
\be\label{FEE}
NRic =\nabla\nabla N,\quad \Delta N=0.
\ee
To deal with geometrically sensitive solutions we demand the metric completeness of $(\Sigma;g)$, (in substitution of geodesic completeness that won't hold here as we will assume $\partial \Sigma\neq \emptyset$ - Metrically complete solutions with empty boundary are flat with constant lapse \cite{MR1806984}). Of great interest are those solutions having compact but non-necessarily connected horizon $\partial \Sigma$, that is for which $N=0$ on the boundary $\partial \Sigma$ that is assumed compact.  Borrowing a terminology often used in theoretical physics, we call such solutions {\it static black holes}. The fundamental Schwarzschild solutions are of course asymptotically flat (AF) black holes, and are the only ones AF by the celebrated uniqueness theorem \cite{Israel},\cite{RobinsonII},\cite{MR876598}. But there are other static black holes that are not AF. A second, somehow trivial family, are the Boosts. As we describe in detail later, they are quotients of the Rindler wedge by two independent translations, and are flat and cylindrical with the lapse growing linearly from the toroidal horizon. Both, the Schwarzschild solutions and the Boosts admit $\Sa$-symmetries. A third family of static $\Sa$-symmetric black holes was found by Myers in 1987 \cite{PhysRevD.35.455}, and independently by Korotkin-Nicolai in 1994 \cite{94aperiodic} using Weyl's reduction of the vacuum static and axisymmetric equations \cite{ANDP:ANDP19173591804}. The M/KN black holes, as we will call them from now on and will be discussed in detail later, have non-trivial topology and are asymptotically Kasner (AK). These three are the only known families of static $\Sa$-symmetric black holes (and in fact of static black holes). Figure \ref{Fig1} depicts the characteristics of each type of solution. In this article we address the problem of classifying $\Sa$-symmetric static black holes and prove that there are no more examples than those of these three families. This is achieved as follows. In Theorem \ref{MAINTHEOREM} we prove that if an axisymmetric static black hole has the topology and the Kasner asymptotic of a M/KN black hole, then it is a M/KN black hole. We devote the whole paper to prove this theorem. Then, combining this result with the theorem proved in \cite{PartI} and \cite{PartII}, stating that any static black hole is either a Schwarzschild black hole, a Boost, or is of Myers/Korotkin-Nicolai type, (that is, it has the same topology and asymptotic as the M/KN black holes), we obtain, as claimed, that any $\Sa$-symmetric static black hole is either a Schwarzschild black hole, a Boost or a Myers/Korotkin-Nicolai black hole. We state this corollary as Theorem \ref{CLASTHEO}.
\vs

\begin{figure}
\centering
\includegraphics[scale=.5]{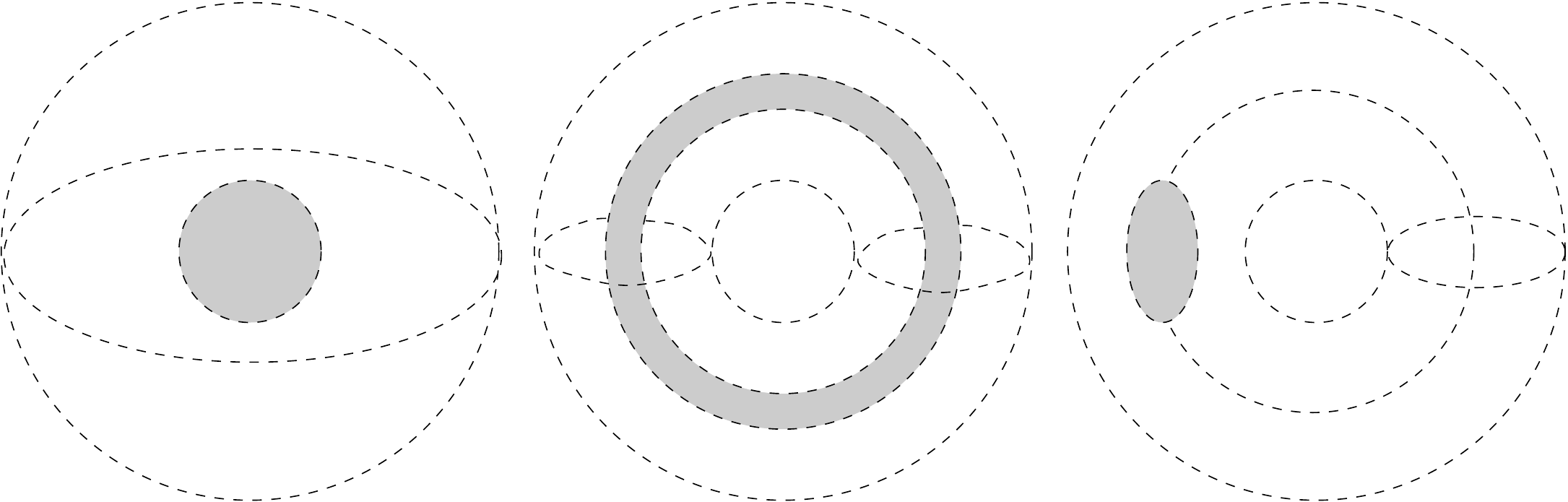}
\caption{A idealization of black holes in each family. From left to right: Schwarschild bhs, Boosts bhs, Myers/Korotkin-Nicolai bhs (with one horizon component in this last case).}
\label{Fig1}
\end{figure}

In the following we describe the technical aspects of the paper. We review again the notion of static data set $(\Sigma;g,N)$, describe the three main families of static black holes earlier mentioned and finally state the main result and comment on the way of proof.

As we said, we will work with static black hole data sets $(\Sigma; g, N)$ that condensate the notion of static black hole at the initial data level (we leave the spacetime picture aside). A static black hole data set consists of a metrically complete smooth orientable three-manifold $(\Sigma;g)$ with non-empty compact boundary, and a smooth function $N$ positive on $\Sigma^{\circ}=\Sigma\setminus \partial \Sigma$ and zero on $\partial \Sigma$ called the lapse, satisfying the static vacuum Einstein equations (\ref{FEE}). Observe that the second equation of (\ref{FEE}) implies that $\Sigma$ is necessarily non-compact (use the maximum principle). The best well known examples are the Schwarzschild black holes given by,
\be
\Sigma=\mathbb{R}^{3}\setminus B(0,2m),\quad g=\frac{1}{1-2m/r}dr^{2}+r^{2}d\Omega^{2},\quad N=\sqrt{1-2m/r},
\ee
with $m>0$ being the mass and the parameter of the family ($B(0,2m)$ is the Euclidean open ball in $\mathbb{R}^{3}$ and radius $2m$).

A data set $(\Sigma;g,N)$ is $\Sa$-symmetric if there is a non-trivial $\Sa$-action $\Sa\times \Sigma\rightarrow \Sigma$ leaving $g$ and $N$ invariant. The $\Sa$-symmetric solution is axisymmetric if the set of fixed points of the action is non-empty. In such case, the set of fixed points are a union of closed disjoint geodesic segments (of finite or infinite length), called the axis of the data. For instance, the Schwarzschild black holes are axisymmetric, where the action is just any $\Sa$-action by rotations on $\mathbb{R}^{3}$. A somehow trivial (but important) family of $\Sa$-symmetric static black holes are the Boosts, defined by the quotient of the data,
\be
\Sigma=\mathbb{R}^{+}_{0}\times \mathbb{R}^{2},\quad g=dx^{2}+dy^{2}+dz^{2},\quad N=x,
\ee
by two linearly independent translations on the $y$-$z$ plane ($y$ and $z$ define the factor $\mathbb{R}^{2}$). Thus, the manifold of a Boost is diffeomorphic to $\Sa\times \Sa\times \mathbb{R}^{+}_{0}$ and the horizon is a two-torus ($\mathbb{R}^{+}_{0}=[0,\infty)$). Boosts are $\Sa$-symmetric as there are two periodic directions.

Let us see the last family, the one that we call the axisymmetric Myers/Korotkin-Nicolai black holes. They were first investigated by Myers in \cite{PhysRevD.35.455} and were rediscovered and further investigated by Korotkin and Nicolai in \cite{94aperiodic}, \cite{KOROTKIN1994229}. As we explain below, M/KN's basic construction uses first Weyl's method to `align' on an axis infinitely many Schwarzschild black holes, and then quotient the resulting static solution to obtain M/KN static black holes. 

Let us begin recalling the basics of how Weyl's reduction is. The presentation will be used in what follows to explain the M/KN construction and will be used later in the proof of the main results. 

Let us assume that $g$ and $N$ are globally of the form\footnote{Weyl showed that around almost every point the static spacetime metric can be put in this form.},
\be
g=e^{-2U}(e^{2\phi}(dz^{2}+d\rho^{2})+\rho^{2}d\varphi^{2}),\quad N=e^{U},
\ee
where the domain of Weyl's coordinates $(\rho,z)$ is $\mathbb{R}^{+}_{0}\times \mathbb{R}$ and where $\varphi$ is the axisymmetric or angular coordinate, obviously $\varphi\in \Sa\sim [0,2\pi)$. Here $U$ and $\phi$ depend only on $(\rho,z)$ and are smooth on $\rho>0$. Accepting this presentation of $g$ and $N$, the static vacuum equations are equivalent to Weyl's equations,
\begin{align}
\label{WEYLCOOR0} & U_{\rho\rho}+\frac{1}{\rho}U_{\rho}+U_{zz}=0,\\
\label{WEYLCOOR20} & \phi_{\rho}=\rho(U_{\rho}^{2}-U_{z}^{2}),\quad \phi_{z}=2\rho U_{\rho}U_{z}.
\end{align}
The equation for $U$ is linear while the second for $\phi$ is solvable by quadratures after having $U$. The basic examples that can be presented in this global form are the Schwarzschild solutions, where in this case $U=U_{S}(\rho,z)$ is,
\be
U_{S}=\ln \bigg(\frac{\sqrt{(z-m)^{2}+\rho^{2}}+\sqrt{(z+m)^{2}+\rho^{2}}-2m}{\sqrt{(z-m)^{2}+\rho^{2}}+\sqrt{(z+m)^{2}+\rho^{2}}+2m}\bigg)
\ee
and $m>0$ is the mass and the parameter of the family. This function is singular just on the segment $\{0\}\times [-m,m]$ over the axis $\rho=0$, which is indeed the projection of the horizon. The function $\phi=\phi_{S}$ turns out to be,
\be
\phi_{S}=\frac{1}{2}\ln\bigg(\frac{\rho^{2}+z^{2}-2m^{2}}{4\sqrt{\rho^{2}+(z+m)^{2}}\sqrt{\rho^{2}+(z-m)^{2}}}\bigg).
\ee
A basic property of the equation (\ref{WEYLCOOR0}) is its linearity: if $U_{1}$ and $U_{2}$ are solutions then $\lambda_{1}U_{1}+\lambda_{2}U_{2}$ is also a solution. Another basic property is the translation invariance along $z$: if $U(\rho,z)$ is a solution, then so is $U(\rho,z+L)$ for any $L$. This turns out to be crucial for the Myers and Korotkin-Nicolai construction. For example, one can sum $U_{S}(\rho,z)$ to $U_{S}(\rho,z+L)$, to obtain a solution that would be interpreted somehow as a `superposition' of two Schwarzschild black holes. If $L>2m$ then the two horizons would be disjoint. The well known problem with this, already known to Bach and Weyl that studied this configuration for the first time \cite{MR2899298}, is that after solving for $\phi$ (assuming $\phi=0$ at infinity) and after reconstructing the spacetime, a strut or conic singularity appears on the axis between the two horizons. Naturally such singularities could be interpreted as a repelling negative singular energy distribution. What Myers and Korotkin-Nicolai showed is that, instead, some infinite superpositions like,
\be\label{UMKN}
U^{\rm MKN}_{L,m}(\rho,z)=U_{S}(\rho,z)+\sum_{n=1}^{\infty}\, \big(U_{S}(\rho,z+nL)+U_{S}(\rho,z-nL)+\frac{4m}{nL}\big)
\ee
(assume $L>2m$), where the counter-terms $4m/nL$ were subtracted to make the series convergent, give, after solving for $\phi$, a spacetime without struts, where in this case infinite horizons are aligned on an axis and are separated from each other at equal lengths (the function $\phi$ can be given as a series). The asymptotic $\rho\rightarrow \infty$ of the solution can be seen to have the Kasner form,
\be
g \sim c_{1}\rho^{\alpha^{2}/2-\alpha}(dx^{2}+d\rho^{2})+c_{2}\rho^{2-\alpha}d\phi^{2},\quad N \sim c_{0}\rho^{\alpha/2},
\ee
where $\alpha=4m/L$ and so $0< \alpha<2$. 
Solutions like this (that depend on $m$ and $L$) are simply connected and for this reason will be called {\it universal Myers/Korotkin-Nicolai static solutions}. They are not static black holes as their boundary, being an infinite union of two-spheres, is not compact. To obtain static black holes we need to take suitable quotients of them. For that purpose we note that they inherit still the $\Sa$-symmetry with the action $R:\Sa\times \Sigma \rightarrow \Sigma$ given by $R(\alpha,(\rho,z,\varphi))=(\rho,z,\alpha+\varphi)$, and have also a discrete 'translational' symmetry $T:\mathbb{Z}\times \Sigma \rightarrow \Sigma$, given by $T(n\times (\rho,z,\varphi)) = (\rho,z+nL,\varphi)$. Both actions $R$ and $T$ commute. So we can quotient simultaneously by a translation and a rotation, obtaining a non-simply connected space with a finite number of spherical horizons, namely, let $\alpha \in [0,2\pi)$ and $h$ a positive integer, then we identify $(\rho,z,\varphi)$ with $(\rho,z+hL,\varphi+\alpha)$. Such quotients are examples of what we call M/KN static black holes. 

\begin{Remark} It seems not have been noted before in the literature that, fixed $m>0$ and $L>2m$ and fixed $h>0$ (i.e.  a number of horizon components), then quotienting with different values of $\alpha\in [0,2\pi)$ results in globally inequivalent solutions. So, the possible quotients are parametrised by $\mathbb{Z}\times [0,2\pi)$.
\end{Remark}  


The particular universal solution just described, with the $U$ given by (\ref{UMKN}), is one possible instance among many. As another instance one could alternate Schwarzschild solutions of masses $m_{1}$ and $m_{2}\neq m_{1}$. More explicitly, letting $L>2m_{1}+2m_{2}$, define $U$ as, 
\be
U=U^{\rm MKN}_{L,m_{1}}(\rho,z)+U^{\rm MKN}_{L,m_{2}}(\rho,z+L/2).
\ee
Then, redoing KN's  argument in \cite{94aperiodic}, one can show that the associated universal spacetime after solving for $\phi$ doesn't have struts either and so is a universal M/KN solution. Several other configurations are possible too. Quotients are taken then as earlier, by a rotation of angle $\alpha\in [0,2\pi)$ and a translation by a multiple of $L$. Without aiming to classify all the possibilities, we just define the family of M/KN static black holes as the set of all possible quotients of all the universal static data set (without struts) that can be constructed using M/KN's method. The manifold $\Sigma$ of any M/KN static black hole, is always diffeomorphic to an open three-torus minus a finite number of open three-balls and the asymptotic of the black hole is Kasner. 

What makes an axisymmetric static black hole data set $(\Sigma;g,N)$ a M/KN static black hole, is whether its universal covering space admits global Weyl's coordinates $(\rho,z)\in\mathbb{R}^{+}_{0}\times \mathbb{R}$, where $U$ takes the form,
\be
U(\rho,z)=\sum_{i=1}^{h}U^{\rm MKN}_{m_{i},L}(\rho,z+z_{i}),
\ee
for some suitable $m_{i}$ and $z_{i}$. This simple characterisation of axisymmetric M/KN static black holes will be used to prove our main Theorem \ref{MAINTHEOREM}.

A static data set is said to be of Myers/Korotkin-Nicolai type (M/KN-type) if it is asymptotically Kasner and the topology of $\Sigma$ is that of a solid torus minus a finite number of open three-balls.
\vs

In this article we will prove that any $\Sa$-symmetric static black hole data set of M/KN-type is indeed a Myers/Korotkin-Nicolai black hole.

\begin{Theorem}\label{MAINTHEOREM} Any $\Sa$-symmetric static black hole data set of M/KN type is indeed a Myers/Korotkin-Nicolai black hole.
\end{Theorem}

Note that $\Sa$-symmetric M/KN-black holes are indeed axisymmetric because their horizons, being two-spheres and $\Sa$-symmetric, have always two fixed points (the poles). Thus, we won't loose generality if inside the statement of the Theorem \ref{MAINTHEOREM} we replace $\Sa$-symmetric by axisymmetric.
\vs

Now, between \cite{PartI} and \cite{PartII} it was proved that any static black hole data set is either a Schwarzschild black hole, a Boost, or is of Myers/Korotkin-Nicolai type. Combining this result with Theorem \ref{MAINTHEOREM} we obtain the following complete classification of axisymmetric static black hole data sets.

\begin{Theorem}[Classification theorem]\label{CLASTHEO} Any vacuum $\Sa$-symmetric static black hole is either a Boost, a Schwarzschild black hole, or a Myers/Korotkin-Nicolai black hole.
\end{Theorem}

To end the introduction let us say a few words on the steps needed to prove Theorem \ref{MAINTHEOREM}. First we prove that the quotient of $\Sigma$ by $\Sa$ is diffeomorphic to $\Sa\times \mathbb{R}^{+}_{0}$. This is achieved mostly by geometric and topological arguments plus the fact, proved in \cite{PartI}, that static black holes have always only one end (that imposes important topological constraints). Then we prove the existence of Weyl's global coordinates $(\rho,z)$. It is shown that the coordinate $\rho$ ranges in $\mathbb{R}^{+}_{0}$ and the periodic coordinate $z$ ranges say in $[0,L)$, interval that we identify with the circle $\Sa_{L}$ of perimeter $L$. So we have the projection,
\be
\pi:\Sigma\rightarrow \mathbb{R}^{+}_{0}\times \Sa_{L},
\ee
given by $p\rightarrow (\rho(p),z(p))$. In this context, the horizon components $\mathcal{H}_{1},\ldots,\mathcal{H}_{h}$ (that make $\partial \Sigma$) project into $h$ disjoint segments $\pi(\mathcal{H}_{1}),\ldots,\pi(\mathcal{H}_{h})$ on $\{0\}\times \Sa_{L}$ of lengths $L_{1},\ldots,L_{h}$ respectively. We identify now the universal cover of $\mathbb{R}^{+}_{0}\times \Sa_{L}$ to $\mathbb{R}^{+}_{0}\times \mathbb{R}$, and we let $\Pi$ be the projection. Then, for each $i$, $\Pi^{-1}(\pi(\mathcal{H}_{i}))$ is an infinite union of disjoint intervals of length $L_{i}$ on $\{0\}\times \mathbb{R}$, where the distance between the centers of two such consecutive intervals is $L$. Let $U_{\pi(\mathcal{H}_{i})}(\rho,z):=U^{\rm MKN}_{L_{i}/2,L}(\rho,z+z_{i})$ be the potential (\ref{UMKN}) as in the M/KN construction, where $m$ is now replaced by $L_{i}/2$ and $z_{i}$ is the $z$-coordinate of the center of any of the intervals in $\Pi^{-1}(\pi(\mathcal{H}_{i}))$. Let finally,
\be
U_{\pi(\mathcal{H}_{1}),\ldots,\pi(\mathcal{H}_{h})}:=\sum_{i=1}^{h}U_{\pi(\mathcal{H}_{i})}.
\ee
We will prove that $U\circ \pi^{-1}\circ \Pi^{-1}$ and $U_{\pi(\mathcal{H}_{1}),\ldots,\pi(\mathcal{H}_{h})}$ differ, if anything, by a constant. This proves that, on the universal cover of the static data set, the potential $U$ is indeed the potential of a M/KN solution. As $\phi$ is found from $U$ uniquely, we deduce that the black hole is a M/KN black hole, finishing thus the proof of Theorem \ref{MAINTHEOREM}.

\section*{Proof of Theorem \ref{MAINTHEOREM}}

The positive reals are denoted by $\mathbb{R}^{+}$, we also use $\mathbb{R}^{+}_{0}=\mathbb{R}^{+}\cup \{0\}$. The closed half-plane is denoted by $\mathbb{R}^{2+}_{0}=\mathbb{R}\times \mathbb{R}^{+}_{0}$. 

We will use a couple of times below that the number of fixed points of any Killing field on a compact orientable Riemannian surface generating a $\Sa$-isometric action is equal to the Euler-characteristic of the surface (use that the zeros of $\xi$ have always index one). So, isometric $\Sa$-actions on tori do not have fixed points, and on spheres have two (the poles). 

\begin{Theorem}\label{NUMBCONCOM} Let $(\Sigma;g,N)$ be a $\Sa$-symmetric static black hole data set of M/KN type. Then the axis is a compact one-manifold.
\end{Theorem}

\begin{proof} If the axis is not compact then there is a divergent sequence $p_{i}(\in \Sigma)$ of fixed points of the $\Sa$-action. Suppose such is the case. As the static data is of M/KN type the asymptotic  is Kasner. Hence, for $i$ sufficiently large, the level set $\{p\in \Sigma: N(p)=N(p_{i})\}$ of $N$ is a $\Sa$-invariant two-torus (level sets are invariants). But non-trivial isometric $\Sa$-actions on two-tori do not have fixed points. We reach thus a contradiction. Hence the axis is compact.
\end{proof}

\begin{Theorem}\label{TTO} Let $(\Sigma;g,N)$ be a $\Sa$-symmetric static black hole data set of M/KN type. Then the quotient manifold $\qM$ of $\Sigma$ by $\Sa$ is homeomorphic to $\Sa\times \mathbb{R}^{+}_{0}$.
\end{Theorem}

\begin{proof} Recall that $\Sigma$ is diffeomorphic to a an open solid torus minus $h>0$ open three-balls. So $\Sigma\sim \hat{\Sigma}\setminus (B_{1}\cup\ldots\cup B_{h})$, where $\hat{\Sigma}$ is an open solid torus and the $B_{i}$'s are open three-balls. Extend $g$ to a $\Sa$-symmetric metric $\hat{g}$ on the whole solid torus $\hat{\Sigma}$, that is, extend $g$ to the balls $B_{i}'s$. Thus $(\hat{\Sigma};\hat{g})$ is now a $\Sa$-symmetric solid torus. Let $\hat{\qM}$ be the quotient of $\hat{\Sigma}$ by $\Sa$, as represented in Figure (\ref{Fig2}). Let $\pi:\hat{\Sigma}\rightarrow \hat{S}$ be the projection. Observe that the axes of $\hat{\Sigma}$ project diffeomorphically onto $\partial \hat{\qM}$. We note too that $\hat{\qM}$ has only one end and is diffeomorphic to $\Sa\times \mathbb{R}^{+}$, namely, outside a compact set $\hat{\qM}$ is diffeomorphic to $\Sa\times \mathbb{R}^{+}$. This is due to the fact that as $(\Sigma;g,N)$ is AK, then the end of $\Sigma$ is foliated by $\Sa$-invariant two-tori (the level sets of $N$), and the $\Sa$-quotient of each invariant torus is diffeomorphic to $\Sa$.

By Theorem \ref{NUMBCONCOM}, $\partial \hat{\qM}$ has a finite number of connected components each one homeomorphic to a circle. If $\partial \hat{\qM}$ has more than one connected component then one can join a point in one component to a point in another component by a connected one-manifold (a closed `segment', see Figure \ref{Fig2}). When lifting that `segment' to $\hat{\Sigma}$ it gives an embedded two-sphere having intersection number one with each of the two connected components of the axis $\pi^{-1}(\partial \hat{\qM})$, where the end points of the segment belonged to. As the intersection number is not zero, such sphere cannot be contractible in $\hat{\Sigma}$. But spheres embedded inside solid tori are always contractible. Hence $\partial \hat{\qM}$ has only one connected component, and so does $\partial \qM$.

On the other hand if the genus of $\qM$ is not zero then there is a closed non-contractible embedded loop such that, if removed from $\qM$, the resulting manifold is still connected (see Figure \ref{Fig2}). When such a loop is lifted to $\Sigma$ it gives an embedded two-torus, say $T$. If such torus is removed (closed) from $\Sigma$ we get still a connected manifold but with two new boundary components. When two copies of such manifold are properly glued together (along the new boundary components) we obtain a double cover $\tilde{\Sigma}$ of $\Sigma$.  When $g$ and $N$ are lifted to $\tilde{\Sigma}$, we obtain a static black hole data set $(\tilde{\Sigma};\tilde{g},\tilde{N})$ with two ends which is not possible as proved in \cite{PartI}.  Hence $\qM$ has genus zero.

Combining the two results above we obtain that $\qM$ is homeomorphic to $\Sa\times \mathbb{R}^{+}_{0}$.

\begin{figure}
\centering
\includegraphics[scale=.45]{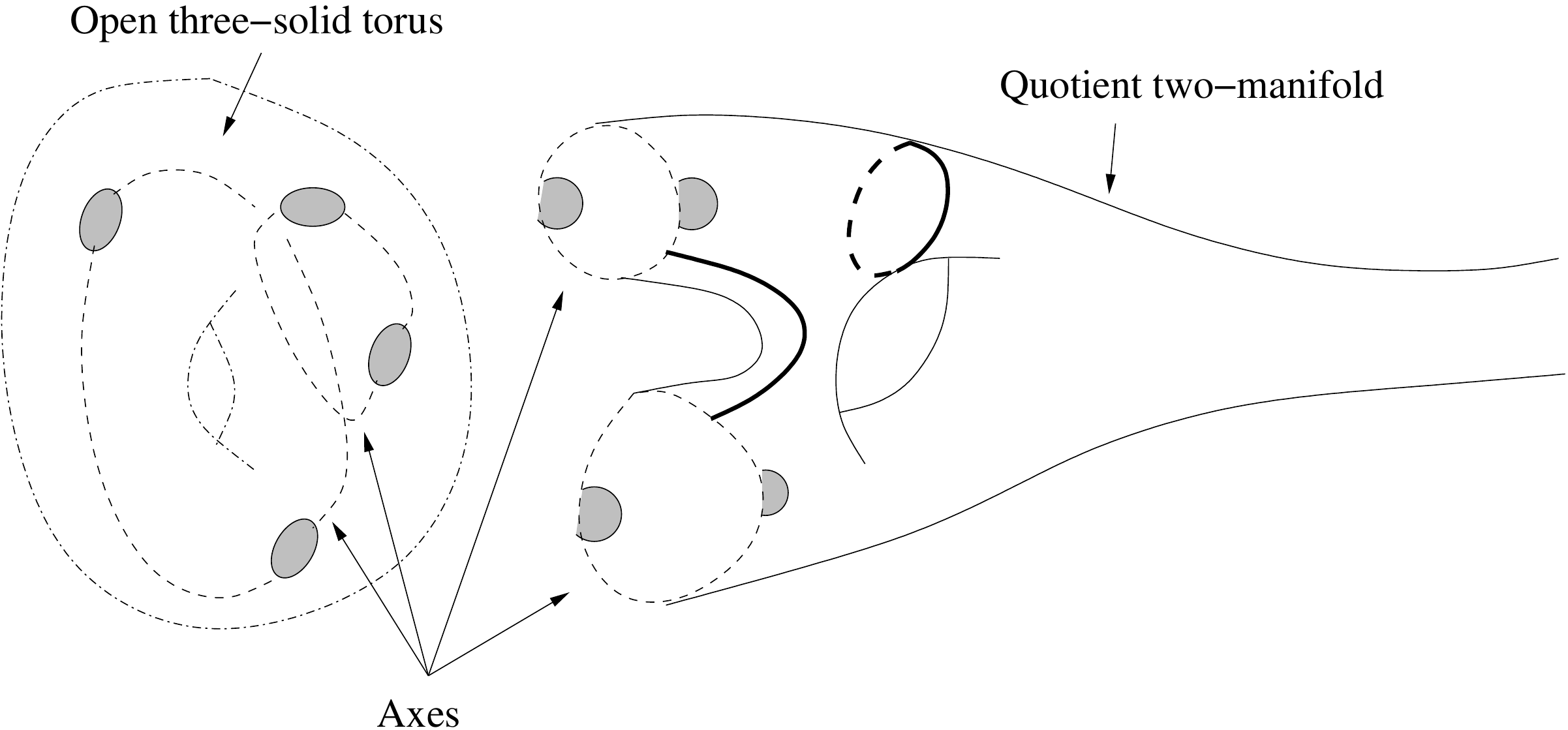}
\caption{Sketch of the construction inside the proof of Theorem \ref{TTO}.}
\label{Fig2}
\end{figure}

\end{proof}

We will make now a few comments on the setup in the quotient space. We will work with it several times later.

Let $\qM\sim \Sa\times \mathbb{R}^{+}_{0}$ be the quotient manifold of a static data set of M/KN type, and let $\Sigma\stackrel{\pi}{\longrightarrow}\qM$ be the projection. The horizons $\mathcal{H}_{1},\ldots,\mathcal{H}_{h}$ project into a set of $h$ disjoint closed `segments' in $\partial \qM$. The axes project diffemorphically into the closure of the complement  in $\partial \qM$ of the projected horizons. The projection $\pi$ is a $\Sa$-principal fibre bundle when restricted to $\Sigma$ minus the axis and the horizons. Furthermore, the connection given by the distribution of two-planes perpendicular to the axisymmetric Killing field is flat (because the distribution is integrable, see \cite{MR757180}), hence there are local charts $W_{i}\subset S\setminus \partial S$ and trivialisations $\chi_{i}: \pi^{-1}(W_{i})\rightarrow W_{i}\times \Sa$ with $\chi(x)=(\pi(x),\varphi_{i}(x))$, having constant transition functions, i.e.  if $\pi(x)\in W_{i}\cap W_{j}$ then $\varphi_{i}(x)=R_{ij}\circ \varphi_{j}(x)$ with $R_{ij}$ a constant rotation on $\Sa$. In other words, the rotational angle $\varphi$ is defined locally up to a constant, but not necessarily globally.        

It is usually convenient to express the metric $g$ as,
\be
g=e^{-2U}(q+\Lambda^{2}d\varphi^{2}),
\ee
where $q$ is the quotient two-metric on $S$, $\varphi$ is the rotational angle (as said, well defined locally up to a constant), $\Lambda=|\partial_{\varphi}|_{g}$ is the $g$-norm of the axisymmetric Killing field $\partial_{\varphi}$. The metric $q$ and the function $U$ are singular on $\partial \qM$, whereas $\Lambda$ is zero there. The static Einstein equations in these variables $q$, $U$ and $\Lambda$, are equivalent to the system,
\begin{align}
\label{EQ1} & \kappa=|\nabla U|^{2},\\
\label{EQ2} & \Delta U + \langle \frac{\nabla \Lambda}{\Lambda},\nabla U\rangle=0,\\
\label{EQ3} & \Delta \Lambda=0,
\end{align}
where $\kappa$ is the Gaussian curvature of $q$. Thus, $\hqg$ has non-negative curvature. Furthermore, due to Anderson's estimate (\cite{MR1806984}, see also another proof in this context in \cite{PartII}),
\be\label{ANDEREST}
|\nabla U|^{2}(p)\leq \frac{c}{\dist^{2}(p,\partial \qM)},
\ee
the curvature $\kappa$ decays to zero at infinity at least quadratically in $d(p,\partial \qM$), (here $\dist$ is the distance function to $\partial \qM$ with respect to $q$). This decay estimate plus the non-negativity of $\kappa$ will be a relevant information when studying the asymptotic. Another fundamental estimate that we will use to study the asymptotic is,
\be\label{GVBOUND}
|\nabla V|^{2}(p)\leq \frac{c}{\dist^{2}(p,\partial \qM)},
\ee
where $V=\ln \Lambda$ (see \cite{PartII}).

The equation (\ref{EQ2}) is nothing else than the projection of the lapse equation $\Delta N=0$ and it can be written in the form ${\rm div}(\Lambda \nabla U)=0$. Hence, by Gauss's theorem, the integral,
\be\label{INTEGRALL0}
\int_{\ell}\Lambda \nabla_{n}Udl,
\ee
over smooth embedded loops $\ell\subset S\setminus \partial S$ isotopic to $\partial \qM$, do not depend on $\ell$ (here $dl$ is the element of length on $\ell$ and $n$ the `outward' unit normal to $\ell$, i.e. pointing inwards to the unbounded component of $\qM\setminus \ell$). This integral is seen easily to be equal the flux of $\nabla N$ on the lift of $\ell$ to $\Sigma$, namely,
\be\label{SURFINT}
\int_{\ell}\Lambda \nabla_{n}Udl=\int_{\pi^{-1}(\ell)}\nabla_{\rm n}NdA,
\ee
where here ${\rm n}$ is the $g$-normal to $\pi^{-1}(\ell)$ and $dA$ is the $g$-element of area. As the surfaces $\pi^{-1}(\ell)$ are connected and enclose the horizons, the surface integrals (\ref{SURFINT}), for any $\ell$, are equal to the flux on the horizons,
\be
\int_{\mathcal{H}}\nabla_{n}N dA.
\ee 
Thus, we have,
\be\label{INTEGRALL0}
\int_{\ell}\Lambda \nabla_{n}Udl=\sum_{i=1}^{i=h}K_{i}A_{i}>0,
\ee
where $K_{i}>0$ and $A_{i}>0$ are the temperature and the area of each horizon respectively (recall that $|\nabla N|=\nabla_{n}N$ is constant over each $\mathcal{H}_{i}$ and is called the temperature of the horizon).
 
\begin{Proposition}\label{BOB} Let $(\qM; \hqg,U,V)$ be the quotient of a static data set of M/KN type. Then $\Lambda$ cannot be uniformly bounded above.
\end{Proposition}

\begin{proof} We already mentioned that the integrals,
\be\label{INTEGRALL}
\int_{\ell}\Lambda \nabla_{n}Ud l,
\ee
over loops $\ell\subset S\setminus \partial S$ isotopic to $\partial \qM$, do not depend on $\ell$, and that the constant they define is positive. We will prove that if $\Lambda$ is bounded above, that is $\Lambda\leq\overline{\Lambda}<\infty$, then (\ref{INTEGRALL}) is necessarily zero, reaching thus a contradiction. We will do so by evaluating the integral over a certain divergent sequence of loops $\ell_{i}$ and proving that it converges to zero (so is zero for all $i$). For that purpose we need to clarify the type of asymptotic that one can have. 

First we recall that the Gauss curvature is $\kappa=|\nabla U|^{2}$. Hence $\kappa$ is non-negative and by (\ref{ANDEREST})  has at least quadratic decay. Standard arguments then show that the asymptotic of $(\qM; \hqg)$ is distinguished by just the area growth. Let us recall this. First fix a loop $\ell_{0}$ isotopic to $\partial \qM$ and let $E$ be the unbounded component of $\qM\setminus \ell$. Let $B(\ell_{0},r)=\{p\in E:\dist (p,\ell_{0})< r\}$ be the `ball' in $E$ of `center' $\ell_{0}$ and radius $r>0$. Then by the Bishop-Gromov monotonicity, the quotient,
\be\label{RESCAREA}
\frac{A(B(\ell_{0},r))}{r^{2}},
\ee
is monotonically decreasing as $r$ increases. Let $\mu\geq 0$ be the limit. The asymptotic of the manifold $(\qM; \hqg)$ is now distinguished by the cases $\mu=0$ and $\mu>0$.

\begin{figure}
\centering
\includegraphics[scale=.45]{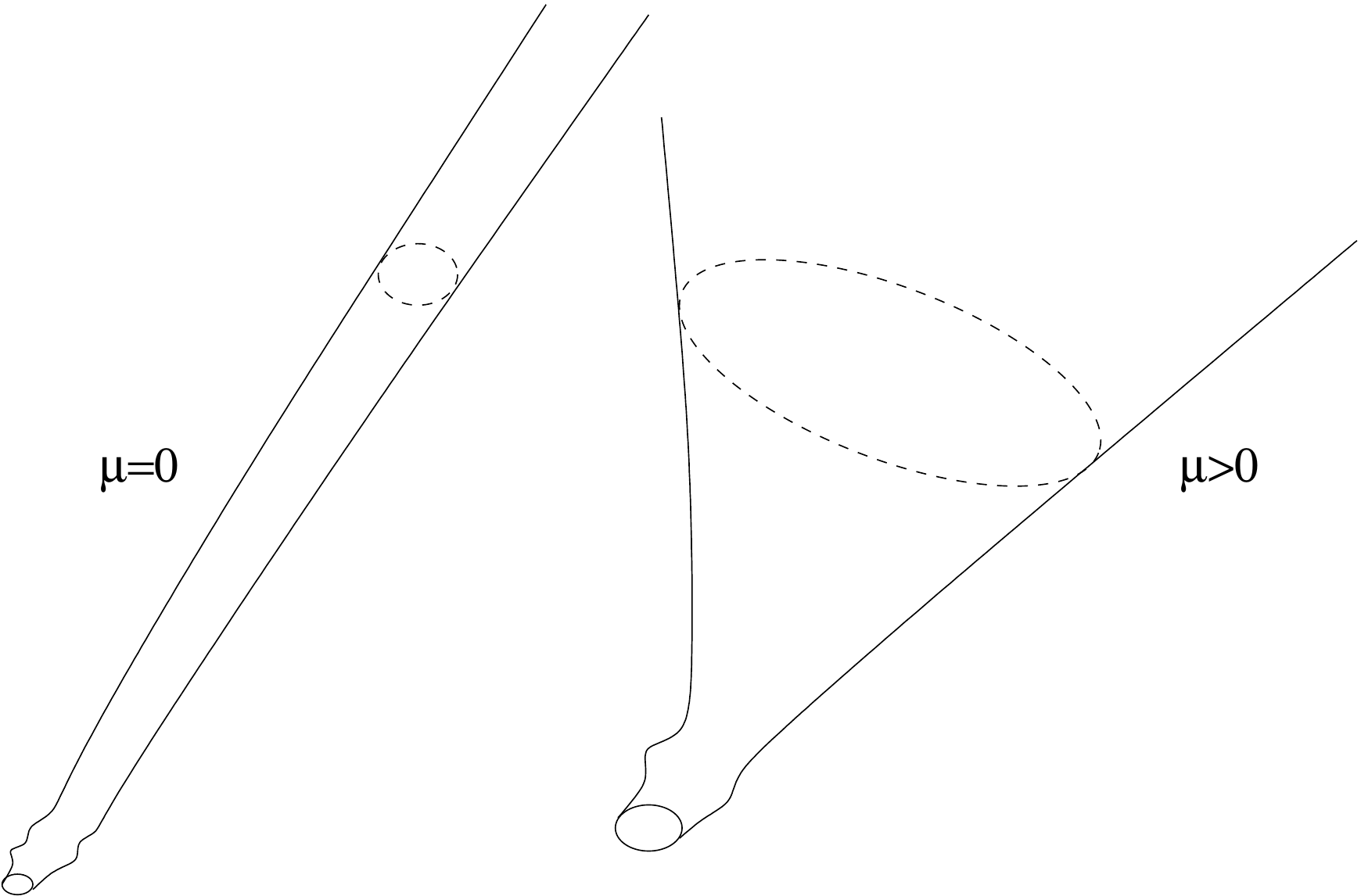}
\caption{Sketch of the asymptotic of $(S;q)$ depending on the area-growth.}
\label{Fig3}
\end{figure}

Let us assume $\mu=0$. Then $\lim_{r\rightarrow \infty} A(B(\ell_{0},r))/r^{2}=0$. For any $r>0$ we consider the annuli $\mathcal{A}(r/2,2r):=B(\ell_{0},2r)\setminus \overline{B(\ell_{0},r/2)}$. As $\mu=0$ the area of these annuli with respect to the scaled metric $\hqg_{r}:=\hqg/r^{2}$, namely
\be
A_{q_{r}}(\mathcal{A}(r/2,2r))=\frac{A(B(\ell_{0},2r))-A(B(\ell_{0},r/2))}{r^{2}}
\ee
tends to zero as $r\rightarrow \infty$. Furthermore by Anderson's estimate, the Gaussian curvature of $\hqg_{r}$ on $\mathcal{A}(r/2,2r)$, namely $\kappa_{q_{r}}=\kappa r^{2}$, is uniformly bounded, that is $\kappa_{\hqg_{r}}=r^{2}\kappa\leq c$, where $c$ independent on $r$. Thus, as $r\rightarrow \infty$ the scaled annuli $(\mathcal{A}(r/2,2r);q_{r})$ collapse in volume (area) with bounded curvature. Therefore their geometry looks like that of thin (finite) cylinders, whose `sections' tend to zero as $r\rightarrow \infty$. As a result one can chose a sequence of loops $\ell_{i}$ embedded in the annuli $\mathcal{A}(r_{i}/2,2r_{i})$ whose $\hqg_{r_{i}}$-length tends to zero, or, equivalently whose $\hqg$ length ${\rm length}(\ell_{i})$ divided by $r_{i}$ tends to zero as $r_{i}\rightarrow \infty$. Therefore,
\be
|\int_{\ell_{i}} \Lambda \nabla_{n}Ud l|\leq \overline{\Lambda}\frac{2c}{r_{i}}{\rm length}_{\hqg}(\ell_{i})\rightarrow 0
\ee
where we have used (\ref{ANDEREST}), ($|\nabla U|(p)\leq 2c/r_{i}$ if $p\in \mathcal{A}(2r_{i},r_{i}/2)$). Thus, the integral (\ref{INTEGRALL}) must be zero.

Let us assume now that $\mu>0$. In this case, the annuli $\mathcal{A}(2r,r/2)$ endowed with the scaled metric $\hqg_{r}$ converge in $C^{\infty}$ to the flat annulus,
\be\label{LIMANNULUS}
\mathcal{A}_{\mathbb{R}^{2}}(2,1/2)=B_{\mathbb{R}^{2}}(0,2)\setminus B_{\mathbb{R}^{2}}(0,1/2),\quad q_{\infty}=d\rho^{2}+\frac{\mu^{2}\rho^{2}}{\pi^{2}}d\phi^{2},
\ee
as $r\rightarrow \infty$, where $(\rho,\phi)$ are polar coordinates on $\mathbb{R}^{2}$ and $B_{\mathbb{R}^{2}}(0,R)$ is the Euclidean ball of center $0$ and radius $R$. (For the $C^{\infty}$ convergence use that the monotonicity of (\ref{RESCAREA}) to $\mu>0$ implies that $\kappa_{q_{r}}=|\nabla U|^{2}_{q_{r}}$ tends to zero and then use standard elliptic estimates on the elliptic system on $\Delta_{q_{r}}U+\langle \nabla V,\nabla U\rangle_{q_{r}}=0$, $\Delta_{q_{r}}V+\langle \nabla V,\nabla V\rangle_{q_{r}}=0$ together with the a priori bound (\ref{GVBOUND})). In particular, the Gaussian curvature $\kappa_{\hqg_{r}}=r_{i}^{2}|\nabla U|^{2}$ restricted over the annuli $\mathcal{A}(2r_{i},r_{i}/2)$ tends to zero as $r\rightarrow \infty$. Take now a sequence of loops $\ell_{i}$ over annuli $\mathcal{A}(2r_{i},r_{i}/2)$ and whose $\hqg_{r_{i}}$ length is less or equal than, say, $\beta>0$, or equivalently ${\rm length}(\ell_{i})/r_{i}\leq \beta$ over the annuli. We can now estimate as follows,
\be
|\int_{\ell_{i}}\Lambda \nabla_{n}Ud l|\leq \overline{\Lambda}(\max_{\ell_{i}}|\nabla U|)\beta r_{i}\rightarrow 0,
\ee
where to deduce the convergence to zero we have used that $\kappa_{q_{r_{i}}}=|\nabla U|^{2}r_{i}^{2}\rightarrow 0$. Thus, the integral (\ref{INTEGRALL}) must again be zero.
\end{proof}

In the paragraph below we use the discussion inside the previous proof to observe that the integral of $\kappa=|\nabla U|^{2}$ over the end of $\qM$ is finite. Namely $\int_{\Omega} |\nabla U|^{2}dA<\infty$, where $\Omega$ is the complement of any open and bounded set in $\qM$ containing $\partial \qM$. This estimate will be used when proving in Proposition \ref{KASSIMP} the Kasner form of $U$ and $V$ in Weyl coordinates.   

According to the discussion of the last Proposition \ref{BOB}, the asymptotic of the manifolds $(\qM; \hqg)$ depends on whether $\mu=0$ or $\mu>0$. Let $r_{i}=2^{i}$ and consider the annuli $\mathcal{A}(2r_{i},r_{i}/2)$ endowed again with the scaled metric $\hqg_{r_{i}}:=\hqg/r^{2}_{i}$. If $\mu=0$ then, for large $i$, the annuli $(\mathcal{A}(2r_{i},r_{i}/2);\hqg_{r_{i}})$ are thin cylinders collapsing (in the Gromov-Hausdorff metric) to a segment as $i\rightarrow \infty$, whereas if $\mu>0$ then they converge to a flat annulus (as $i\rightarrow \infty$). In the later case, we can find loops $\ell_{i}$ in $\mathcal{A}(2r_{i},r_{i}/2)$, dividing the annuli in two, and whose $\hqg_{r_{i}}$-length tends to $2\mu$ and whose mean curvature $\theta_{i}$ tends to one (the limit loop is $\rho=1$ in (\ref{LIMANNULUS})). By Gauss-Bonnet, if we let $\Omega(\ell_{i_{0}},\ell_{i})$ be the finite closed cylinder enclosed by $\ell_{i_{0}}$ and $\ell_{i}$ then,
\be\label{WALALA}
\int_{\Omega(\ell_{i_{0}},\ell_{i})} \kappa dA=\int_{\ell_{i}}\theta_{i}dl-\int_{\ell_{i_{0}}}\theta_{i_{0}}dl
\ee
Observe that the product $\theta_{i}dl$ is scale invariant, so we can compute the integral $\int_{\ell_{i}}\theta_{i}dl$, using $\theta_{i}$ and $dl$ with respect to $q_{r_{i}}$ that we know remain uniformly bounded as $i\rightarrow \infty$. Then, as $i\rightarrow \infty$, the right hand side remains bounded. Thus $\kappa$, hence also $|\nabla U|^{2}$, have finite integrals on the unbounded region of $\qM\setminus \ell_{i_{0}}$. The same holds when $\mu=0$ but a few extra words must be said. Having in mind the formula (\ref{WALALA}) we need to select the loops $\ell_{i}$ in such a way that $\int_{\ell_{i}}\theta_{i}dl$ remains uniformly bounded as $i\rightarrow \infty$. We can select the $\ell_{i}$ such that the $q_{r_{i}}$-length tends to zero, but so far we do not have control on the mean curvature with respect to $q_{r_{i}}$ (the geometry is collapsing). One can however take finite covers of the annuli (to have the injectivity radius bounded from above and from below away from zero) converging in $C^{\infty}$ to a $\Sa$-symmetric data and then just take $\ell_{i}$ as loops that lift to loops $\tilde{\ell}_{i}$ converging to an $\Sa$-symmetric loop of finite mean curvature. The loops $\ell_{i}$ then have bounded mean curvature because they are just the mean curvatures of the lifts. Hence, also in this case, the integral of $|\nabla U|^{2}$ on the unbounded component of $\qM\setminus \ell_{i_{0}}$ is bounded.
\vs

We will prove now that $\Lambda$ tends to infinity over the end of $\qM$, namely, for any divergent sequence $p_{i}\rightarrow \infty$, we have $\Lambda(p_{i})\rightarrow \infty$.

\begin{Proposition} Let $(\qM; \hqg,U,V)$ be the quotient of a static data set of M/KN type. Then $\Lambda\rightarrow \infty$ at infinity.
\end{Proposition}

\begin{proof} Let $r_{i}=2^{i}$ and consider the annuli $\mathcal{A}(2r_{i},r_{i}/2)$ endowed with the scaled metric $\hqg_{r_{i}}=\hqg/r_{i}^{2}$. As explained earlier, the type of asymptotic of the manifolds $(\qM; \hqg)$ depends on whether $\mu=0$ or $\mu>0$ but in either case for $i$ large enough we can find loops $\ell_{i}$ embedded in $\mathcal{A}(2r_{i},r_{i}/2)$ and isotopic to $\partial \qM$, whose $\hqg_{r_{i}}$-lengths are uniformly bounded. We will use such loops below.

Now, recall that on $\mathcal{A}(2r_{i},r_{i}/2)$ we have $|\nabla \ln \Lambda|_{\hqg_{r_{i}}}\leq c$. Let $p^{i}_{\max}$ and $p^{i}_{\min}$ be the points on $\ell_{i}$ where $\Lambda$ achieves its max and min respectively, that is, $\Lambda(p^{i}_{\max})=\max_{\ell_{i}}\{\Lambda\}$ and the same for $\min$. Let $\alpha(t)$ be a curve parameterising the part of $\ell_{i}$ from $p^{i}_{\min}$ to $p^{i}_{\max}$. Then,
\be\label{BOUNDED}
\ln \bigg(\frac{\max_{\ell_{i}}\{\Lambda\}}{\min_{\ell_{i}}\{\Lambda\}}\bigg)=|\int_{t_{0}}^{t_{1}}\langle \nabla \ln \Lambda,\alpha'\rangle_{\hqg_{r_{i}}} dt|\leq \int_{t_{0}}^{t_{1}} |\nabla \Lambda|_{q_{r_{i}}}|\alpha'|_{q_{r_{i}}}dt \leq c L.
\ee 
This is an important uniform bound that we will use below.

For any $i<j$ denote by $\mathcal{U}_{i,j}$ the compact region enclosed by $\ell_{i}$ and $\ell_{j}$. Now, if there is a sequence $j_{k}\rightarrow \infty$ as $k\rightarrow \infty$, for which $\min_{\ell_{j_{k}}}\{\Lambda\}$ remains uniformly bounded, then $\max_{\ell_{j_{k}}}\{\Lambda\}$ remains uniformly bounded too by (\ref{BOUNDED}). Then, by the maximum principle we have,
\be
\max_{\mathcal{U}_{j_{0}j_{k}}}\{\Lambda\}\leq \max\{\max_{\ell_{j_{0}}}\{\Lambda\},\max_{\ell_{j_{k}}}\{\Lambda\}\},
\ee
and thus we conclude that $\Lambda$ remains uniformly bounded contradicting Proposition \ref{BOB}. Hence, $\min_{\ell_{j}}\{\Lambda\}$ tends to infinity as $j$ tends to infinity. Now, again by the maximum principle we have,
\be\label{FINAL}
\max_{\mathcal{U}_{j,j+1}}\{\Lambda\}\geq \min\{\min_{\ell_{j}}\{\Lambda\},\min_{\ell_{j+1}}\{\Lambda\}\}\rightarrow_{j\rightarrow \infty}\infty.
\ee
Now, the concatenation of the regions $\mathcal{U}_{j,j+1}$ covers $\qM$ except a bounded region. Hence we conclude from (\ref{FINAL}) that $\Lambda\rightarrow \infty$ at infinity.
\end{proof}

The next proposition says that $\Lambda$ can be taken as the standard Weyl coordinate $\rho$. 

\begin{Proposition} Let $(\qM; \hqg,U,V)$ be the quotient of a static data set of M/KN type. Then $\Lambda$ can be chosen as a global harmonic coordinate and we can write,
\be\label{COORDSYS}
\hqg=e^{2\phi}(d\rho^{2}+dz^{2}),
\ee
where we have changed notation to the usual $\rho=\Lambda$ and where $z$ is the (periodic of period $L$) coordinate harmonically conjugated to $\rho$.
\end{Proposition}

\begin{proof} We know that $\Lambda\rightarrow \infty$ at infinity and that $\Lambda|_{\partial \qM}=0$. Then, the pre-image of any regular value of $\Lambda$ (naturally greater than zero) is necessarily a finite set of circles. As $\Lambda$ is harmonic, the maximum principle implies that none of such circles can enclose a disc. Hence, every circle is isotopic to $\partial \qM$. If there is more than one circle, then any two of them must enclose an annulus which is again ruled out by the maximum principle. Thus, the pre-image of any regular value of $\Lambda$ is a circle isotopic to $\partial \qM$. Fix two regular values $\Lambda_{2}> \Lambda_{1}>0$. Let $\mathcal{A}_{12}$ be the annulus enclosed by $\Lambda^{-1}(\{\Lambda_{1}\})$ and $\Lambda^{-1}(\{\Lambda_{2}\})$. We claim that there are no critical points of $\Lambda$ in $\mathcal{A}_{12}$. To see this, observe that, as $\Lambda$ is harmonic and analytic, the critical points $\{c_{i}\}$ are isolated and of positive (integer) index\footnote{To see this argue as follows. Let $z=x+iy$ be a complex coordinate around a critical point of $\Lambda$ ($z=0$ at the critical point). Let $\Gamma$ be a harmonic conjugate to $\Lambda$, so that $f(z)=\Lambda+i\Gamma$ is analytic. By Cauchy-Riemann $d\Gamma|_{z=0}=0$, hence $z=0$ is also a critical point of $f$. Then, near $z=0$ we have $f(z) \sim c_{1} + c_{2}z^{m}$ for some $m\geq 2$. Moreover by Cauchy-Riemann we have $\partial_{z} f(z)=\partial_{x} \Lambda + i\partial_{x}\Gamma=\partial_{x}\Lambda-i\partial_{y}\Lambda$. Thus, $\nabla \Lambda$ is directly identified to the conjugate of $\partial_{z} f(z)$ (as a vector field). Hence the index at any critical point is always equal to the index of a field $\overline{z}^{n}$ with $n\geq 1$, hence positive.}, and that, by Poincar\'e-Hopf, the sum of the index of the critical points in $\mathcal{A}_{12}$ must be zero. As this is valid for any $\Lambda_{1}$ and $\Lambda_{2}$, it follows that $\Lambda$ does not have critical points.

Make $\rho:=\Lambda$. Now observe that, because $\rho$ is harmonic, the one form $*d\rho$ ($*$ is the $\hqg$-Hodge-star) is closed. Furthermore $*d\rho(\nabla^{i} \rho)=0$ (indeed $*d\rho(\nabla^{i}\rho)=\epsilon_{ij}\nabla^{i}\rho\nabla^{j}\rho=0$ where $\epsilon_{ij}$ is antisymmetric). Thus, $*d\rho$ descends to a one form in the quotient manifold of $\qM$ by the integral curves of the gradient of $\rho$, which is a circle. Define thus $z$ by integrating $dz=*d\rho$ with the initial condition that $z$ is identically zero at one fixed integral curve. 
\end{proof}

In the coordinate $(\rho,z)$, the equations for $\phi$ and $U$ reduce to the Weyl equations,
\begin{align}
\label{WEYLCOOR} & U_{\rho\rho}+\frac{1}{\rho}U_{\rho}+U_{zz}=0,\\
\label{WEYLCOOR2} & \phi_{\rho}=\rho(U_{\rho}^{2}-U_{z}^{2}),\quad \phi_{z}=2\rho U_{\rho}U_{z}.
\end{align}

We will now study the asymptotic of the fields $U$ and $\phi$ as $\rho\rightarrow \infty$, which of course will have the Kasner form,
\be
q=q_{0}\rho^{2u_{0}^{2}}(d\rho^{2}+dz^{2}) + o_{e}(\rho),\quad U=U_{0}+u_{0}\ln \rho+o_{e}(\rho),\quad \phi=\phi_{0}+u_{0}^{2}\ln \rho+o_{e}(\rho).
\ee
where $q_{0}$, $U_{0}$ and $\phi_{0}$ are real constants, and where $|o_{e}(\rho)|\leq \alpha e^{-\beta \rho}$ for some $\alpha>0$ and $\beta>0$. The Kasner form could in principle be obtained by studying directly the reduction of the $\Sa$-symmetric static data set of M/KN type that we are dealing with and that, by its same definition, is assumed to have Kasner asymptotic. Interesting enough however, the Kasner form of the asymptotic of $U$ and $\phi$ will be deduced just from the form (\ref{COORDSYS}) of $q$, from the Weyl equations, and from the fact that the integral of $|\nabla U|^{2}$ over the end of $\qM$ is finite, as was discussed earlier. From these facts we will obtain first a representation of $U$ in series and from it we will read off the asymptotic of $U$ itself and of $\phi$. This is the content of the next proposition.

\begin{Proposition}\label{KASSIMP} Let $(\qM; \hqg,U,V)$ be the quotient of a static data set of M/KN type. Then the asymptotic of $U$ and $\phi$ is Kasner. 
\end{Proposition}
\begin{proof} We will work with the coordinates $(\rho,z)$. Without loss of generality assume that $z\in [0,2\pi)$. To obtain the expression of $U$ and $\phi$ when the range of $z$ is not $[0,2\pi)$ rather $[0,L)$, just make the transformations $z\rightarrow (L/2\pi)z$ and $\rho\rightarrow (L/2\pi)\rho$.

Separating variables, the solution to (\ref{WEYLCOOR}) is always expressible in series as,
\be
U=U_{0}+u_{0}\ln \rho+\sum_{n\in \mathbb{Z}\setminus \{0\}}(u_{n}K_{0}(n\rho)+v_{n}I_{0}(n\rho))e^{niz},
\ee
where $u_{n}=u_{-n}^{*}, v_{n}=v_{-n}^{*}$ for each $n$ and where $I_{0}(x)$ and $K_{0}(x)$ are the modified Bessel functions of first and second kind respectively\footnote{$I_{0}(x)$ and $K_{0}(x)$ are linearly independent solutions of the ODE, $x^{2}y''+xy'-x^{2}y=0$.}. These functions have the asymptotic
\be\label{AYE}
K_{0}(x)\approx \sqrt{\frac{\pi}{2}}\frac{\diss e^{-x}}{\sqrt{x}},\qquad I_{0}(x)\approx \frac{1}{\sqrt{2\pi}}\frac{\diss e^{x}}{\sqrt{x}}.
\ee
If we show that $v_{n}=0$ for all $n$ then a direct computation using (\ref{WEYLCOOR2}) and the standard properties of the modified Bessel functions brings
\begin{align}
\phi-\phi_{0}= &u_{0}^{2}\ln \rho +2 \sum_{n\in \mathbb{Z}\setminus \{0\}}  u_{0}u_{n}K_{0}(n\rho)e^{inz}\\
& - \sum_{n,m\in \mathbb{Z}\setminus \{0\}}  u_{n}u_{m}\frac{nm}{n+m}\rho\big(K_{0}(m\rho)K_{1}(n\rho)+K_{0}(n\rho)K_{1}(m\rho)\big)e^{i (n+m)z}
\end{align}
Hence, if $v_{n}=0$ for all $n$ then the static end converges exponentially in $\rho$ to the Kasner solution
\be\label{KASSSI}
q=q_{0}\rho^{2u_{0}^{2}}(d\rho^{2}+dz^{2}),\quad U=U_{0}+u_{0}\ln \rho,\quad \phi=\phi_{0}+u_{0}^{2}\ln \rho
\ee
We show now that indeed $v_{n}=0$ for all $n$.

Let $\mathcal{A}_{01}=\{(\rho,z):\rho_{0}\leq \rho\leq \rho_{1}\}$. Now observe that
\be
\int_{\mathcal{A}_{01}} U_{\rho}e^{inz}\, d\rho dz = 2\pi(u_{n}K_{0}(n\rho)+v_{n}I_{0}(n\rho))\bigg|_{\rho_{0}}^{\rho_{1}}
\ee
Now, if $v_{n}\neq 0$, then we can use the asymptotic expansion (\ref{AYE}) and deduce that the integral in the left hand side has exponential growth in $\rho_{1}$. On the other hand, as shown earlier, ends have bounded total curvature, therefore as $\kappa=|\nabla U|^{2}$ we get,
\be
\int_{\rho\geq \rho_{1}}U_{\rho}^{2}\, d\rho dz<\infty
\ee
(use that $\kappa=|\nabla U|^{2}$ and that the expression $|\nabla U|^{2}dA$ is conformal invariant). Next use Cauchy-Schwarz to get
\be
\big|\int_{\mathcal{A}_{01}} U_{\rho}e^{inz}\, d\rho dz \big|\leq \big(\int_{\mathcal{A}_{01}}U_{\rho}^{2}\, d\rho dz\big)^{1/2}(2\pi(\rho_{1}-\rho_{0})\big)^{1/2}
\ee
But the left hand side grows faster than the right hand side and we reach a contradiction. Hence $v_{n}=0$ for all $n$.
\end{proof}

Let's recall what we have so far. The function $U(\rho,z)$ is smooth on the quotient manifold $\Sa_{L}\times \mathbb{R}^{+}$, where $z$ is the coordinate of the factor $\Sa_{L}$ (a circle of length $L$) and $\rho$ is the coordinate of the factor $\mathbb{R}^{+}$, and satisfies $U_{zz}+U_{\rho\rho}+\rho^{-1}U_{\rho}=0$. Let $\pi:\Sigma\rightarrow \Sa_{L}\times \mathbb{R}^{+}_{0}$ be the projection. Then the projection of the horizons components, $\pi(\mathcal{H}_{1}),\ldots,\pi(\mathcal{H}_{h})$, are $h$ disjoint closed segment on $\Sa_{L}\times \{0\}$ of $z$-lengths $L_{1},\ldots,L_{h}$ respectively (i.e. if $\pi(\mathcal{H}_{i})=\{(z,0):z_{i}^{-}\leq z\leq z_{i}^{+}\}$, then $L_{i}=z_{i}^{+}-z_{i}^{-}$). To prove Theorem \ref{MAINTHEOREM}, it remains to show that $U$ is actually equal to the M/KN potential $U_{\pi(\mathcal{H}_{1}),\ldots,\pi(\mathcal{H}_{h})}$ up to a constant. We will show that in what follows. Recall from the introduction that $U_{\pi(\mathcal{H}_{1}),\ldots,\pi(\mathcal{H}_{h})}$ is the sum of the potentials $U_{\pi(\mathcal{H}_{i})}$,
\be
U_{\pi(\mathcal{H}_{1}),\ldots,\pi(\mathcal{H}_{h})}=\sum_{i=1}^{h}U_{\pi(\mathcal{H}_{i})}.
\ee

To simplify notation let us make $U_{\pi(\mathcal{H}_{1}),\ldots,\pi(\mathcal{H}_{h})}=\tilde{U}$. We will show that, 
\be
w=U-\tilde{U},
\ee
is actually a constant function. To achieve that we will use an integral identity and the behaviour of $w$ near the segments $\pi(\mathcal{H}_{1}),\ldots,\pi(\mathcal{H}_{h})$ and near infinity. Let us prove this finishing thus the proof of Theorem \ref{MAINTHEOREM}.

First, as $w$ satisfies $w_{zz}+w_{\rho\rho}+\rho^{-1}w_{\rho}=0$, then $div(w\rho dw)=\rho |dw|^{2}=\rho(w_{z}^{2}+w_{\rho}^{2})$. Integrating using Gauss's theorem we obtain the integral identity,
\be
\int_{\rho_{0}\leq \rho\leq \rho_{1}} \rho(w_{z}^{2}+w_{\rho}^{2})dzd\rho=\int_{\rho=\rho_{1}}\rho ww_{\rho}dz-\int_{\rho=\rho_{0}}\rho ww_{\rho}dz
\ee
for any $0<\rho_{0}<\rho_{1}$. We will show that the boundary terms on the right hand side tend to zero as $\rho_{0}\rightarrow 0$ and $\rho_{1}\rightarrow \infty$, thus proving that $w$ is constant. 

We treat first the case $\rho_{1}\rightarrow \infty$ and then we treat the case $\rho_{0}\rightarrow 0$. 

Case $\rho_{1}\rightarrow \infty$. We begin recalling that $U$ and $\tilde{U}$ have the following expansions as $\rho\rightarrow \infty$,
\begin{gather}\label{UNO}
U=U_{0}+u_{0}\ln \rho +o(1/\rho^{2}),\quad U_{\rho}=\frac{u_{0}}{\rho}+o(1/\rho^{3}),\\
\label{DOS}
\tilde{U}=\tilde{U}_{0}+\tilde{u}_{0}\ln \rho+o(1/\rho^{2}),\quad \tilde{U}_{\rho}=\frac{\tilde{u}_{0}}{\rho}+o(1/\rho^{3}).
\end{gather}
(we know that the decay in the $o$'s is faster indeed but that won't be needed). Thus, if $u_{0}=\tilde{u}_{0}$, then $w=w_{0}+o(1/\rho^{2})$ and $w_{\rho}=o(1/\rho^{3})$. It follows that $\rho ww_{\rho}=o(1/\rho^{2})$ so,
\be
\int_{\rho=\rho_{1}}\rho ww_{\rho}dz\rightarrow 0
\ee
as $\rho_{1}\rightarrow \infty$. We will prove that indeed,
\be
u_{0}=\tilde{u}_{0}=\frac{L_{1}+\ldots+L_{h}}{L}<1
\ee
We will do so by showing that the integrals,
\be\label{INTEGRALSVALUE}
\int_{\rho=\rho_{0}}\rho U_{\rho}dz,\quad \int_{\rho=\rho_{0}}\rho \tilde{U}_{\rho}dz,
\ee
that we know do not depend on $\rho_{0}$, are both equal to $L_{1}+\ldots+L_{h}$. Let us begin calculating the first integral. To start we note that $\rho$ has the expression, $\rho=|\partial_{\varphi}|_{g}e^{U}$, and that the factor $|\partial_{\varphi}|_{g}$ is a smooth non-singular function on the three-manifold $\Sigma$. Then, around any point $(0,z_{0})$ on a horizon component $\pi(\mathcal{H}_{i})$ except the pole points, (i.e. if $\pi(\mathcal{H}_{i})=\{z^{-}_{i}\leq z\leq z^{+}_{i}\}$ then $z^{-}_{i}<z_{0}<z^{+}_{i}$), we have the expansion,
\be\label{UEST}
U=\ln \rho + O(\rho,z),\quad \partial_{\rho} U=\frac{1}{\rho}+O(\rho,z).
\ee
As we will see now, these expansions allow us to show that the first integral of (\ref{INTEGRALSVALUE}) takes indeed the value $L_{1}+\ldots+L_{h}$. We already know that the integral is equal to,
\be\label{A}
\sum_{j=1}^{j=h}\int_{\mathcal{H}_{j}}\nabla_{n}NdA,
\ee
where in this expression each summand can be obviously calculated as,
\be\label{B}
\lim_{\epsilon\rightarrow 0}\int_{\mathcal{H}_{j}^{\epsilon}}\nabla_{n} NdA,
\ee
where $\mathcal{H}_{j}^{\epsilon}=\pi^{-1}(\{(0,z):z^{-}_{j}+\epsilon<z<z^{+}_{j}-\epsilon\})$ is the horizon minus small discs at the poles. The advantage of doing that is that now we can write without risk,
\be
\int_{\mathcal{H}_{j}^{\epsilon}}\nabla_{n} NdA = \lim_{\rho_{0}\rightarrow 0} \int_{z^{-}_{j}+\epsilon<z<z^{+}_{j}-\epsilon} \rho_{0} U_{\rho}(\rho_{0},z)dz = z^{+}_{j}-z^{-}_{j}-2\epsilon,
\ee
where to obtain the limit as $\rho_{0}\rightarrow 0$ we have used the expansions (\ref{UEST}). Plugging this back into (\ref{B}) and taking the limit $\epsilon\rightarrow 0$, and then plugging the result into (\ref{A}) we obtain,
\be
\int_{\rho=\rho_{0}}\rho U_{\rho}dz=L_{1}+\ldots+L_{h},
\ee
as wished. The calculation of the second integral in (\ref{INTEGRALSVALUE}) is done in exactly in the same fashion, except that now we work separately over each M/KN data set constructed out of each potential $U_{\pi(\mathcal{H}_{i})}$ to show that $\int_{\rho=\rho_{0}}\rho U_{\pi(\mathcal{H}_{i}),\rho}dz=L_{i}$ and hence that the second integral in (\ref{INTEGRALSVALUE}) is equal to $L_{1}+\ldots+L_{h}$. 

Case $\rho_{0}\rightarrow 0$. Calculating the integral in this case is in a sense simpler, except for the integral near the poles that has to be properly justified. Let $\epsilon>0$ be sufficiently small. Let $Z_{\epsilon}$ be the set of $z\in \Sa_{L}$ such that, the distance to any end point of the intervals $\pi(\mathcal{H}_{i})$ is larger than $\epsilon$. Let $\mathcal{H}_{\epsilon}$ be the set of those points of $Z_{\epsilon}$ lying in $\pi(\mathcal{H})$, and let $\mathcal{A}_{\epsilon}$ be the set of those points of $Z_{\epsilon}$ not lying in $\pi(\mathcal{H})$. Then, we can write,
\begin{align}
\int_{\rho=\rho_{0}}\rho ww_{\rho} dz = & \int_{\rho=\rho_{0}, z\in \mathcal{H}_{\epsilon}}\rho ww_{\rho}dz + \int_{\rho=\rho_{0}, z\in \mathcal{A}_{\epsilon}} \rho ww_{\rho}dz\ + \\
\label{FINALII} & \int_{\rho=\rho_{0}, z\in (\mathcal{H}_{\epsilon}\cup \mathcal{A}_{\epsilon})^{c}}\rho ww_{\rho}dz
\end{align}
We fix $\epsilon>0$ small and let $\rho_{0}\rightarrow 0$. Let us explain the limit of each of the three integrals on the right. On points over the horizon that are not the poles, both, $U$ and $\tilde{U}$, have the same expansion (\ref{UEST}), so $w=O(\rho,z)$ and $w_{\rho}=O(\rho,z)$. Hence the first integral on the right hand side tends to zero. On points over the axis both $U$ and $\tilde{U}$ are smooth and thus $w$ and $w_{\rho}$ are also smooth. Hence the limit of the second integral on the right hand side is zero too. For the third integral, that is (\ref{FINALII}), we show below that for $\epsilon$ small $\rho ww_{\rho}$ remains bounded as $\rho_{0}\rightarrow 0$ (the bound does not depend on $\epsilon$, indeed tends to zero as $\epsilon$ tends to zero). Hence, the third integral remains bounded by $C\epsilon$ es $\rho_{0}\rightarrow 0$. Together with the other two limits explained earlier, we obtain that the limsup as $\rho_{0}\rightarrow 0$ of the integral,
\be
\int_{\rho=\rho_{0}}\rho w w_{\rho} dz,
\ee
is less or equal than $C\epsilon$, for all $\epsilon>0$ small. Hence the limit is indeed zero as wished. 

To finish, let us show that $\rho ww_{\rho}$ on the integrand of (\ref{FINALII}) remains bounded. The key is to change coordinates around each of the poles (assume $z=0$ at the given pole) from the harmonic coordinates $(\rho,z)$ to the harmonic coordinates $u\geq 0,v\geq 0$, defined by $(v^{2}-u^{2})/2=z$, $uv=\rho$. Near the pole, the axis is $\{u=0, v\geq 0\}$ while the horizon is $\{v=0,u\geq 0\}$. The coordinates $(u,v)$ are correlated with normal coordinates on $\Sigma$. Indeed, let $p\in \partial \Sigma$ be the pole. Let $\gamma(s):[0,\mu]\rightarrow \Sigma$, ($\mu$ small), be a geodesic parameterised by arc-length $s$ and tangent to $\partial \Sigma$ at $p=\gamma(0)$. From each $\gamma(s)$ let $\delta_{s}(t):[0,\nu]\rightarrow \Sigma$ be a geodesic perpendicular to $\partial \Sigma$ at $\delta_{s}(0)=\gamma(s)$. In this way $(s,t)\in [0,\mu]\times [0,\nu]$ are normal smooth coordinates around the pole $p$. Near the pole, the axis is $\{s=0,t\geq 0\}$ while the horizon is $\{t=0,s\geq 0\}$. The change of variables $(s,t)\rightarrow (u,v)$ is smooth. Thus, the lapse functions $N=e^{U}$, which is smooth as a function of $(s,t)$, is also smooth as a function of $(u,v)$, and therefore has a expansion,
\be
N=cv + p(u,v),
\ee
where $c$ is a positive constant, $p$ is smooth, and furthermore $p(u,0)=0$, $\partial_{v} p(u,0)=0$. Similarly\footnote{This time define $(s,t)$ on a M/KN solution.}, $\tilde{N}=e^{\tilde{U}}$ is a smooth function of $(u,v)$, has a expansion $\tilde{N}=\tilde{c}v + \tilde{p}(u,v)$ and $\tilde{p}(u,0)=0$, $\partial_{v}\tilde{p}(u,0)=0$. Hence, using Taylor expansions, we can factor out $p$ and $\tilde{p}$ as $p(u,v)=v^{2}q(u,v)$ and $\tilde{p}(u,v)=v^{2}\tilde{q}(u,v)$ with $q$ and $\bar{q}$ smooth. Thus, near the pole we have $U=\ln c+\ln v + \ln (1+vq)$ and $\tilde{U}=\ln \tilde{c}+\ln v +\ln (1+v\tilde{q})$. Hence $w=\ln c/\tilde{c}+\ln (1+vq/c)/(1+v\tilde{q}/\tilde{c})$ is smooth near the pole and is zero at it. Furthermore a simple calculation gives,
\be
\rho\partial_{\rho}w=u\bigg(\frac{v^{2}}{u^{2}+v^{2}}\bigg)\partial_{u}w+v\bigg(\frac{u^{2}}{u^{2}+v^{2}}\bigg)\partial_{v}w
\ee
As $w$ is smooth and the two factors in parenthesis are bounded by one we deduce that $\rho ww_{\rho}=w\rho w_{\rho}$ is also bounded.
 
This finishes proving that $U$ differs from $\tilde{U}$ by a constant, completing the proof of the Theorem \ref{MAINTHEOREM}. 
 
\bibliographystyle{plain}
\bibliography{Master} 
 
\end{document}